\documentclass[compress,authoryear,5p]{elsarticle}
\usepackage[round]{natbib}
\usepackage{picins}

\journal{Automatica}

\usepackage{amsmath}
\usepackage{natbib}
\newtheorem{theorem}{\textbf{Theorem}}
\newtheorem{lemma}{\textbf{Lemma}}
\newtheorem{example}{\textbf{Example}}
\newtheorem{corollary}{\textbf{Corollary}}
\newtheorem{remark}{\textbf{Remark}}
\newtheorem{definition}{\textbf{Definition}}
\newtheorem{problem}{\textbf{Problem}}
\newtheorem{proposition}{\textbf{Proposition}}
\newenvironment{proof}{{\noindent{\bf \emph{Proof:}}}\quad}{\hfill $\square$\par}
\usepackage{multirow} 
\usepackage{xcolor}
\usepackage{multirow}
\usepackage{geometry}
\usepackage{url}
\usepackage{subfigure}
\usepackage{fancyhdr}
\usepackage{amsmath}
\usepackage{multirow}
\usepackage{amssymb }
\usepackage{color}
\usepackage{graphics} 
\usepackage{graphicx}
\usepackage{algorithm} 
\usepackage{algorithmic} 
\usepackage{setspace}
\begin{document}

\begin{frontmatter}

\title{Structural Controllability of Networked Relative Coupling Systems}

\author{Yuan Zhang, Yuanqing Xia, and Dihua Zhai}
\address{School of Automation, Beijing Institute of Technology, Beijing, China\\~Email: $\emph{\{zhangyuan14,xia\_yuanqing,zhaidih\}@bit.edu.cn}$}

\fntext[myfootnote]{This work was supported in part by the China Postdoctoral Innovative Talent
Support Program (No. BX20200055), the China Postdoctoral Science Foundation (No. 2020M680016), the National Natural Science Foundation of
China (NNSFC, No. 62003042), and the State Key Program of NNSFC (No. 61836001).}

\begin{abstract} This paper studies the controllability of networked relative coupling systems (NRCSs), in which subsystems are of fixed high-order linear dynamics and coupled through relative variables depending on their neighbors, from a structural perspective. The purpose is to explore conditions for subsystem dynamics and network topologies under which for almost all weights of the subsystem interaction links, the corresponding numerical NRCSs are controllable, which is called structurally controllable. Three types of subsystem interaction fashions are considered: 1) each subsystem is single-input-single-output (SISO), 2) each subsystem is multiple-input-multiple-output (MIMO), and the weights for all channels between two subsystems are identical, and 3) each subsystem is MIMO, but different channels between two subsystems can be weighted differently.
{We show that all parameter-dependent modes of the NRCSs are generically controllable under some necessary connectivity conditions. We then derive necessary and/or sufficient conditions for structural controllability depending on subsystem dynamics and network topologies' connectivity in a decoupled form for all the three interaction fashions.} We also extend our results to handle certain subsystem heterogeneities and demonstrate their direct applications on some practical systems, including the mass-spring-damper system and the power network.
\end{abstract}


\end{frontmatter}

\section{Introduction}
Relative coupling/sensing, {namely, coupling through relative information/sensing relative variables, rather than the absolute,} is ubiquitous in real-world dynamic systems, ranging from natural systems like thermal propagating systems, liquid flow systems  and car-following traffic systems \citep{Modern_Control_Ogata}, to human-made ones like unmanned aerial vehicle (UAV) formation systems \citep{olfati2004consensus} and extremely large telescope control systems \citep{sarlette2014control}. {Networked relative coupling system {(NRCS) is a class of networked systems whose subsystems interact with each other through relative variables depending on their neighbors \citep{olfati2004consensus,hamdipoor2019partitioning,zhang2014upper}}}, which belongs to the so-called diffusively coupled networks in the literature \citep{zhang2014upper,Menara2020StabilityCF}. There have been many active research topics on NRCSs, such as consensus \citep{olfati2004consensus}, synchronization \citep{Menara2020StabilityCF}, stability \citep{hamdipoor2019partitioning}, etc. Among them, a fairly fundamental property, namely, controllability/observabi-lity, has also attracted many researchers' interest.
It is widely accepted that this property is not only theoretically significant, as itself is often related to both algebraical and topological properties of networks \citep{R.Am2009Controllability}, but also relevant to other important system aspects, such as stabilization, the existence of optimal controllers, and designing formation protocols \citep{Modern_Control_Ogata}.

Concerning the controllability of NRCSs, many works have focused on the controllability of networks with Laplacian related dynamics. In the field of multi-agent systems (MASs), many studies resolve this issue using spectra analysis of Laplacian matrices or graph-theoretic tools \citep{R.Am2009Controllability,aguilar2015graph,zhang2014upper}. Particularly, the controllability of consensus-based MASs is studied in \cite{R.Am2009Controllability} and \cite{zhang2014upper} using  (almost) equitable partitions and graph automorphism. Some graph-theoretic characterizations for the controllability of Laplacian-based leader-follower systems are reported in \cite{aguilar2015graph}. {However, most of these works do not consider the situations that nodes constituting the networks may have high-order dynamics and that each node may be multi-input-multi-output (MIMO). Thus the interactions among them may not be described by graphs with scalar-weighted edges.}

On the other hand, significant efforts have also been devoted to networks of high-order linear systems (with general coupling mechanisms). Relevant works include \cite{L.Wa2016Controllability}, \cite{hao2018further} and \cite{xue2018modal} on networks of identical systems,  \cite{zhou2015on} and \cite{Y_Zhang_2016} on networks with heterogeneous subsystems, and \cite{chapman2014controllability} on networks of networks. These works are built upon completely deterministic system models seeking to find relationships between controllability and network topologies as well as subsystem dynamics, {and most of their results are rank conditions based on the PBH test. Structural controllability, a notion focusing on the controllability in the generic sense that does not rely on the precise system parameters, has also been adopted in network studies \citep{Composability,Liu2019AGC,zhang2019structural,Commault2019Generic}. For example, \cite{Composability} explored structural controllability on composite systems with an emphasis on the distributed verification, and \cite{Liu2019AGC} worked on systems satisfying a so-called `binary' parameterization.} Recently, the structural controllability of networked systems is considered in \cite{zhang2019structural} and \cite{Commault2019Generic}, where the subsystem dynamics are partially or completely fixed under the assumption that subsystem interaction weights can take values independently. {Note that such an assumption might prevent their results from being directly applied to NRCSs, as there exist zero row sum constraints in the associated Laplacian matrices.}

{In this paper, motivated by some real-world systems we adopt a practical model for a class of NRCSs, in which subsystems are of fixed identical high-order dynamics, but the interaction weights among them are unknown.} We study the structural controllability of the NRCSs, with the purpose of finding conditions for subsystem dynamics and the network topologies under which the associated numerical systems are controllable for almost all values of interaction weights. Three types of subsystem interaction
fashions are considered, including: 1) each subsystem is
single-input-single-output (SISO), 2) each subsystem is
MIMO with equally weighted interaction channels, and 3) each subsystem is MIMO, but the interaction channels between two subsystems can be weighted differently. {For each of the three types of interaction fashions, we give necessary and/or sufficient conditions for structural controllability depending on subsystem dynamics (algebraic conditions) and the network topologies (graph-theoretic conditions) in a decoupled form.  Notably, for the SISO subsystem case, our results naturally generalize \cite{goldin2013weight} and \cite{Kazemi2018Structural} where the consensus-based networks of single integrators are considered. A design procedure is also given to construct interaction weights for controllable NRCSs with given SISO subsystems (Section \ref{section_SISO}). For the last two interaction fashions, we also show that under some necessary connectivity conditions, all parameter-dependent modes of the NRCSs are generically controllable (Sections \ref{section_single} and \ref{MIMO_section}). Thus, the structural controllability verification problems collapse to verifying generic ranks at some fixed modes of subsystems.  To the best of our knowledge, the results for the second case are among the early attempts to give {\emph{graph-theoretical conditions}} for structural controllability in which the indeterminates can have high-rank coefficient matrices (along with \citealt{menara2018structural} etc.), and those for the third case are to understand structural controllability of networks with multiplex links \citep{tuna2016synchronization,Lombana2020DistributedIO}.} We finally extend our results to handle certain subsystem heterogeneities, which is illustrated by some typical practical
 systems (Section~\ref{section_V}).

  {Interested readers are referred to \cite{zhang2020structural} for an extension of the third interaction fashion to NRCSs with undirected network topologies. In contrast to \cite{zhang2020structural}, where the weights of multiplex links between two subsystems must form a {\emph{full}} matrix, in this paper, they can form a diagonal matrix, which in principle enables representing an arbitrarily prescribed zero-nonzero structure by suitable matrix transformations. Additionally, this paper reveals some new insights into the role of network topologies's connectivity in the existence of parameter-dependent uncontrollable modes.}

{\emph{Notations:}} For a set, $|\cdot |$ denotes its cardinality.
$\sigma(M)$ denotes the set of eigenvalues of matrix $M$, and $\otimes$ denotes the Kronecker product. ${\bf diag}\{X_i|_{i=1}^n\}$ denotes the block diagonal matrix whose $i$th diagonal block is $X_i$, while ${\bf col}\{X_i|_{i=1}^n\}$ the matrix with its $i$th row block being $X_i$. Denote the set of all $m\times n$ matrices by ${\mathbb M}^{m\times n}$. A directed graph (digraph for brevity) is denoted by ${\cal G}=({\cal V},{\cal E})$ where ${\cal V}$ is the vertex set and ${\cal E}\subseteq {\cal V}\times {\cal V}$ is the edge set. The set of edges of ${\cal G}$ is also denoted by $E({\cal G})$.

\section{Problem formulation} \label{section_II}
{\subsection{Motivating example}
{Before introducing the formal system model, we introduce a motivating example. Consider a typical mass-spring-damper system shown in Fig. \ref{damping_vehicle} \citep{Modern_Control_Ogata,zhang2019structural}, which consists of $N$ subsystems.} For the $i$th subsystem, let $x_i$ be the displacement of the mass, $m_i$, $k_i$, and $\mu_i$ be the mass, the constants of the spring and the damper, respectively, and  $u_i$ be the force imposed on the mass. Let $x_{i1}=x_i$ and $x_{i2}=\dot x_i$. Then the dynamics for the $i$th mass (subsystem) can be written as the following state-space model
{\small
\[\begin{array}{l}
\left[ \begin{array}{l}
{{\dot x}_{i1}}\\
{{\dot x}_{i2}}
\end{array} \right] = {\left[ {\begin{array}{*{20}{c}}
0&1\\
0&0
\end{array}} \right]}\left[ \begin{array}{l}
{x_{i1}}\\
{x_{i2}}
\end{array} \right] + \sum \limits_{j=i-1,i+1} l_{ij}^{[1]}b_1c_1\left[ {\begin{array}{*{20}{c}}
{{x_{j1}} - {x_{i1}}}\\
{{x_{j2}} - {x_{i2}}}
\end{array}} \right]\\
+\sum \limits_{j=i-1,i+1} l_{ij}^{[2]}b_2c_2\left[ {\begin{array}{*{20}{c}}
{{x_{j1}} - {x_{i1}}}\\
{{x_{j2}} - {x_{i2}}}
\end{array}} \right]+b\frac{{{u_i}}}{{{m_i}}},
\end{array}\]}where $b=b_1=b_2=[0,1]^{\intercal}$, $c_1=[1, 0]$, $c_2=[0, 1]$, $l^{[1]}_{i,i-1}=k_i/m_i$, $l^{[1]}_{i,i+1}=k_{i+1}/m_i$, $l^{[2]}_{i,i-1}=\mu_i/m_i$, and $l^{[2]}_{i,i+1}=\mu_{i+1}/m_i$.

{In the above model, subsystems are coupled through relative information. The intrinsic dynamics of each subsystem when isolated are known from physical modeling (and probably identical). The unknown parameters (parameters $m_i$, $k_i$, and $\mu_i$ for each subsystem) are reflected in the interaction weights. Many practical networked systems share similar characteristics (e.g., {the connected tank system} and the power network; see Section \ref{section_V}).}

\begin{figure}
  \centering
  \includegraphics[width=3.0in]{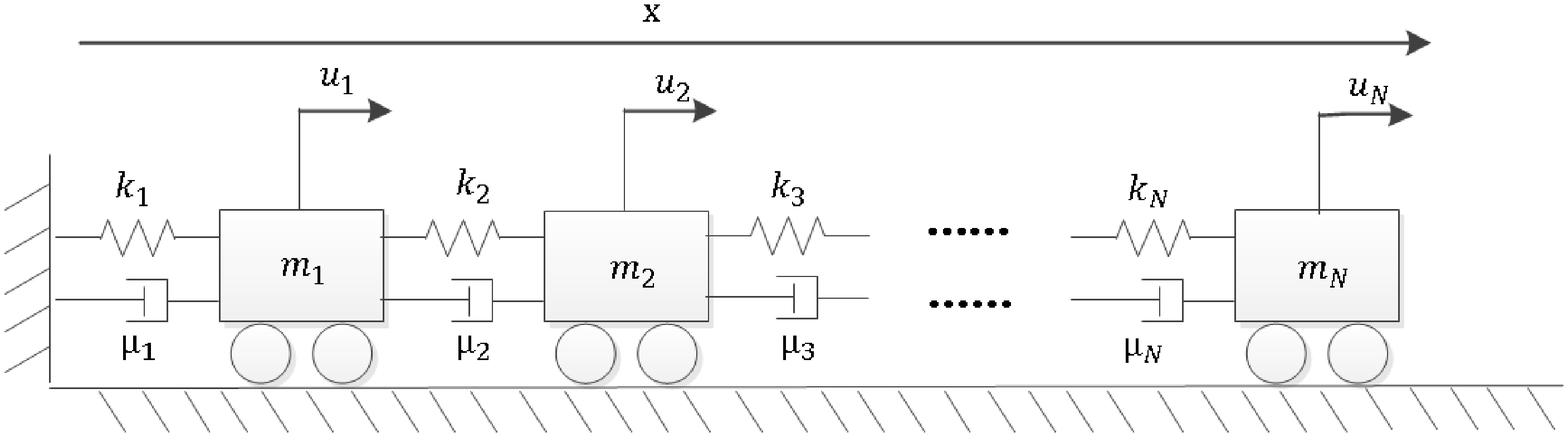}
  \caption{The mass-spring-damper system \citep{Modern_Control_Ogata,zhang2019structural}}\label{damping_vehicle}
\end{figure}

\subsection{Problem statement}

{ In this paper, an NRCS formally refers to a system in which each subsystem interacts with others through relative information (i.e., the difference between state or output variables) depending on its neighbors over a network. More precisely, let ${\cal G}_{\rm sys}=({\cal V}_{x}, {\cal E}_{xx})$ be a digraph without self-loops describing the subsystem interaction topology of an NRCS, which consists of $N$ subsystems, with ${\cal V}_{x}=\{1,...,N\}$, and $(i,j)\in {\cal E}_{xx}$ if the $j$th subsystem is directly influenced by the $i$th one.}  The $i$th subsystem, denoted by $S_i$, $i\in\{1,...,N\}$, has the following dynamics
\begin{equation} \label{sub_dynamic}
\dot x_i(t)=Ax_i(t)+Bv_i(t)
\end{equation}
where $A\in {\mathbb R}^{n\times n}$, $B\doteq[b_1,...,b_r]$ with $b_j\in {\mathbb R}^{n}$ for $j=1,...,r$, and $x_i(t)\in {\mathbb R}^{n}$ and $v_i(t)\in {\mathbb R}^{r}$ are the state and input of $S_i$, respectively. The input $v_i(t)$ may contain both subsystem interactions and the external control inputs, whose $k$th component, denoted by $v_{ik}(t)$, is expressed~as
\begin{equation} \label{sub_interaction}
v_{ik}(t)\!\!=\!\!\sum \limits_{j=1,j\ne i}^N l^{[k]}_{ij}c_k(x_j(t)-x_i(t))+\delta_i u_{ik}(t),
\end{equation}for $k=1,...,r$. Here, $u_{ik}(t)$ is the $k$th component of the external input $u_i(t)\in {\mathbb R}^{r}$, i.e., $u_i(t)\doteq [u_{i1}(t),...,u_{ir}(t)]^{\intercal}$, $\delta_i\in \{0,1\}$ with $\delta_i=1$ meaning that $S_i$ is directly controlled by $u_i(t)$, and $\delta_i=0$ the contrary; $c_k\in {\mathbb R}^{1\times n}$ is the output vector generating the $k$th linear combination of the state difference $x_j(t)-x_i(t)$ (namely, internal output), and $l^{[k]}_{ij}\in {\mathbb R}$ is the weight imposed on $c_k(x_j(t)-x_i(t))$. Define $C\doteq [c_1^{\intercal},...,c_r^{\intercal}]^{\intercal}$. We will say that $(A,B,C)$ are parameters describing the subsystem intrinsic dynamics. For each $k\in \{1,...,r\}$, $l^{[k]}_{ij}\ne 0$ only if $(j,i)\in {\cal E}_{xx}$ ($i\ne j$).
Let $l^{[k]}_{ii}=-\sum \nolimits_{j=1,j\ne i}^N l^{[k]}_{ij}$.  {Define the (weighted) Laplacian matrix associated with ${\cal G}_{\rm sys}$ as $L_k=[-l^{[k]}_{ij}]$, which means that the entry in the $i$th row and $j$th column of $L_k$ is $-l^{[k]}_{ij}$}.  {Let ${\cal I}_u=\{i: \delta_i\ne 0\}$, and $\Delta=[{\bf e}^{N}_i|_{i\in {\cal I}_u}]$, where ${\bf e}^{N}_i$ represents the $N$ dimensional column vector whose $i$th entry is $1$ and the rest are zero.} Define $u(t)={\bf col}\{u_i(t)|_{i\in {\cal I}_u}\}$, and $x(t)={\bf col}\{x_i(t)|_{i=1}^N\}$. Then the lumped state-space representation of the NRCS (\ref{sub_dynamic})--(\ref{sub_interaction}) is
\begin{equation} \label{lump_ss}
\dot x(t)=A_{\rm sys}x(t) + B_{\rm sys}u(t),
\end{equation}with
\begin{equation} \label{lump_pp_MIMO} A_{\rm sys}=I_N\otimes A-\sum \nolimits_{k=1}^r L_k\otimes b_kc_k, B_{\rm sys}=\Delta \otimes B.\end{equation}

{\begin{definition} \label{definition_SC}
Suppose that the parameters $(A,B,C)$ are known for each subsystem and ${\cal G}_{\rm sys}$, $\Delta$ are given. The {NRCS} (\ref{sub_dynamic})--(\ref{sub_interaction}) is said to be structurally controllable, if there exists a set of values for $\{l^{[k]}_{ij}\}_{(j,i)\in {\cal E}_{xx}}^{k=1,...,r}$, such that the associated numerical system is controllable.
\end{definition}}

{Using the algebraic variety arguments (cf. \citealt{Murota_Book}), it is easy to prove that}, if system (\ref{sub_dynamic})--(\ref{sub_interaction}) is structurally controllable, then for almost all values for $\{l^{[k]}_{ij}\}^{k=1,...,r}_{(j,i)\in {\cal E}_{xx}}$ except for a set of Lebesgue measure zero in the  parameter space, the corresponding numerical systems are controllable. In other words, controllability is a {\emph {generic property}} for the NRCS (\ref{sub_dynamic})--(\ref{sub_interaction}) \citep{generic}.

\begin{remark}\label{known_sub} In many practical scenarios, subsystems parameters $(A,B,C)$ might be either known accurately from physical modeling (the most common case is the high-order integrators; cf. the motivating example, the UAV formation systems [\citealp{olfati2004consensus}]) or easily accessible from system identification. The subsystem interaction weights, on the contrary,  might be harder to obtain due to the geographical distance between subsystems or variants of parameters dominating the interaction channels \citep{zhang2019structural,Lombana2020DistributedIO}. Moreover, parameter variants for subsystem intrinsic dynamics sometimes could be decoupled and put into the interaction weights (see Section \ref{section_V}). Based on  these  observations,\! and also,\! for ease of exposition, the parameters $(A,B,C)$ are \!assumed \!to \!be \!known \!herein.
\end{remark}
\vspace*{-0.7em}
{{ \begin{remark}\label{compare_inf} \cite{Composability} studied the structural controllability of networked systems where both subsystem dynamics and their interconnections are described by structured matrices, i.e., matrices whose entries are either fixed zero or unknown free parameters, in the sense of Lin's theory \citep{C.T.1974Structural}, rather than Definition \ref{definition_SC}. While their adopted model can cover general networks of linear systems, and the obtained results are noticeable owing to simple graphical representations and efficient computations, it is expected that with more {\emph{available}} prior knowledge on subsystem dynamics and the involved parameter dependencies  taken into account, more accurate results may be obtained \citep{generic}. A further discussion on possible inclusion of the case where not all subsystem parameters are exactly known will be given in Corollary \ref{corollary_het} of Section \ref{section_V}.
\end{remark}}}

In this paper, arising from observations on some practical systems, we will consider three types of subsystem interaction fashions depending on the subsystem inputs/outputs. They are: (a) the SISO case, i.e., $r=1$ meaning that each subsystem is SISO; (b) the MIMO via equally weighted channels, where $r>1$ and $L_1=\cdots = L_r$; and (c) the MIMO via differently weighted channels, where $r>1$ and there is no parameter dependency among $L_1,...,L_r$ except that they share the same zero-nonzero patterns. See Fig. \ref{couple_channel} for illustrations. The first two cases are called as scalar-weighted networks in \cite{tuna2016synchronization}, and have been the research focus in most existing literature \citep{zhang2014upper,L.Wa2016Controllability,hao2018further,xue2018modal,Commault2019Generic}. { The last case is inspired by the observation that in some practical networks, different internal outputs of~subsystems may represent diverse physical variables, and are therefore transmitted by channels with corresponding parameters (cf. the motivating example).} {Networks in which two nodes/agents are connected by different types of links are recently called `multiplex networks', which form a subset of multilayer networks \! \citep{Lombana2020DistributedIO}.} The main problem considered in this paper is formulated as follows.

\begin{problem} \label{prob1} Suppose that parameters $A,B,C$ of each subsystem, as well as ${\cal G}_{\rm sys}$ and $\Delta$,  are known for the NRCS (\ref{sub_dynamic})--(\ref{sub_interaction}). {Under each of the three aforementioned subsystem interaction fashions}, verify whether the NRCS (\ref{sub_dynamic})--(\ref{sub_interaction}) is structurally controllable or not.
\end{problem}

Note that there are nonzero constants like $A, B$, and $C$ in $(A_{\rm sys}, B_{\rm sys})$, as well as the zero row sum constraints on the Laplacian matrices $L_k|_{k=1}^r$. Moreover, each indeterminate (free parameter) may have a coefficient matrix with rank larger than one in fashion (b).  The traditional Lin's theory \citep{C.T.1974Structural}, as well as the results of \cite{zhang2019structural} and \cite{Commault2019Generic}, cannot be directly adopted to Problem~\ref{prob1}.

\begin{figure}
  \centering
  \includegraphics[width=2.1in]{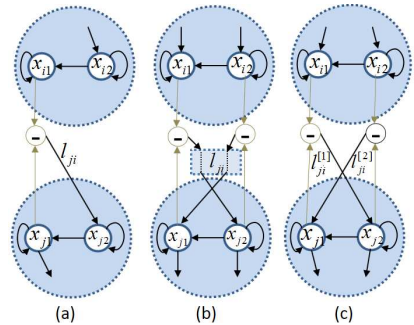}
  \caption{Three types of interaction fashions considered in this paper. From the left to the right: SISO, MIMO via equally  weighted channels, and MIMO via differently weighted channels. }\label{couple_channel}
\end{figure}

\section{Preliminaries} \label{preliminary}
{\bf Definitions and notations in graph theory:} For a digraph ${\cal G}=({\cal V},{\cal E})$ with $N$ vertices, a path from vertex $v_i$ to vertex $v_j$ is a sequence of edges $(v_i,v_{i+1})$, $(v_{i+1},v_{i+2})$,$\cdots$,\\ $(v_{j-1},v_{j})$ where each edge belongs to $\cal E$.  A cycle is a path from a vertex to itself. A spanning tree of $\cal G$ is a subset of $N-1$ edges that form a tree \citep{Geor1993Graph}. {{With a little abuse of terminology, we say $\cal G$ has a spanning tree with the topological order $(v_{k_1},...,v_{k_N})$, if in this tree the parent of $v_{k_{i}}$ is among vertices $\{v_{k_1},...,v_{k_{i-1}}\}$, $\forall i\in \{2,...,N\}$. }} 

For the NRCS (\ref{sub_dynamic})--(\ref{sub_interaction}), a vertex $i\in {\cal V}_{x}$ is called a {driving vertex}, if $i\in {\cal I}_u\doteq \{i: \delta_i\ne 0\}$. Let ${\cal U}=\{u_i: i\in {\cal I}_u\}$, and define a digraph $\bar {\cal G}_{\rm sys}=({\cal V}_{x}\cup {\cal U}, {\cal E}_{xx}\cup {\cal E}_{ux})$, where ${\cal E}_{ux}=\{(u_i,i), i\in {\cal I}_u\}$.  We say a vertex $i$ is input-reachable,  if there exists at least one path from a vertex $u\in {\cal U}$ and ending at $i$ in $\bar {\cal G}_{\rm sys}$. If every vertex $i\in {\cal V}_{x}$ is input-reachable, we say the network topology (or $\bar {\cal G}_{\rm sys}$) is {\emph {globally input-reachable}}. {{It is clear that global input-reachability of $\bar {\cal G}_{\rm sys}$ means that, ${\cal G}_{\rm sys}$ can be decomposed into a collection of (vertex) disjoint trees, of which each is rooted at one driving vertex.}}

Given a matrix $[H, P]$ where $H\in {\mathbb M}^{n\times n}$ and $P\in {\mathbb M}^{n\times m}$, we will use ${\cal G}_{\rm aux}( H, P)$ to denote the auxiliary graph associated with $[ H, P]$, which is defined analogically to $\bar {\cal G}_{\rm sys}$ associated with $(L_k,\Delta)$ ($k\in \{1,...,r\}$). More precisely, ${\cal G}_{\rm aux}( H, P)=({\cal V}_{H}\cup {\cal V}_P, {\cal E}_{HH}\cup {\cal E}_{PH})$, where ${\cal V}_H=\{v_1,...,v_n\}$, ${\cal V}_P=\{z_1,...,z_m\}$, ${\cal E}_{HH}=\{(v_i,v_j):  H_{ji}\ne 0\}$ and ${\cal E}_{PH}=\{(z_i,v_j):  P_{ji}\ne 0\}$. We say a vertex $v_i\in {\cal V}_H$ is input-reachable if there is a path from one vertex of ${\cal V}_{P}$ and ending at $v_i$. Global input-reachability of $G_{\rm aux}(H,P)$ is defined similarly to that of $\bar {\cal G}_{\rm sys}$. \!A cycle of ${\mathcal G}_{\rm aux}(H, P)$ is input-reachable if every vertex of this cycle is input-reachable.

{\bf Structural controllability with parameter dependencies:} { Structural controllability with parameter dependencies has been discussed in \cite{Morse_1976}, \cite{Hosoe1984StructuralCA}, \cite{zhang2019structural}, \cite{Liu2019AGC} and the references therein, based on which the following definitions are introduced. Let $p\doteq (p_1,...,p_m)$ with each $p_i$ an indeterminate. Let $R[s,p]$ be the set of polynomials of the variables $s,p_1,...,p_m$ with real coefficients, and $R[s,p]^{n_1\times n_2}$ the set of $n_1\times n_2$ matrices whose entries belong to $R[s,p]$. $R[s]$, $R[p]$, $R[s]^{n_1\times n_2}$ and $R[p]^{n_1\times n_2}$ are defined similarly. Let $\deg(a(s,p))$ denote the degree of $s$ for $a(s,p)\in R[s,p]$. For $a(s,p)\in R[s,p]$, $\alpha(s)\in R[s]$ is an $s$-factor of $a(s,p)$, if $a(s,p)$ is dividable by $\alpha(s)$ {\emph{and}} $\deg(\alpha(s))\ge 1$.  $\alpha(s,p)\in R[s,p]$ is called an $(s,p)$-factor of $a(s,p)\in R[s,p]$, if $\alpha(s,p)$ divides $a(s,p)$ {\emph{and}} $\alpha(s,p)$ has no $s$-factors.  For $M(s,p)\in R[s,p]^{n\times q}$ with $q\ge n$, denote by $\Gamma (M(s,p))$ the greatest common divisor of the determinants of all $n\times n$ submatrices of $M(s,p)$.

\begin{definition} \citep{Hosoe1984StructuralCA} \label{p-parameterized} Let $(A(p),B(p))$ be a plant parameterized by $p$ with the state transition matrix $A(p)\in R[p]^{n\times n}$ and the input matrix $B(p)\in R[p]^{n\times r}$. Zeros of the $(s,p)$-factors (resp. $s$-factors) of $\det(sI-A(p))$ are called parameter-dependent modes  (resp. fixed modes) of the plant.
\end{definition}

\begin{definition}\citep{Hosoe1984StructuralCA,zhang2019structural} \label{p-parameterized2} Consider the plant $(A(p),B(p))$ in Definition \ref{p-parameterized}. This plant is said to have parameter-dependent uncontrollable modes (resp. fixed uncontrollable modes), if $\Gamma([sI-A(p),B(p)])$ has $(s,p)$-factors (resp. $s$-factors).
\end{definition}}

From the above definitions, a parameter-dependent (resp. fixed) uncontrollable mode is the eigenvalue of $A(p)$ that depends on $p$ (resp. is independent of $p$) and is always uncontrollable. The set of parameter-dependent (resp. fixed) uncontrollable modes is a subset of parameter-dependent (fixed) modes. By the PBH test and properties of algebraic variety, {\emph{$(A(p),B(p))$ is structurally controllable, if and only if there exist neither parameter-dependent uncontrollable modes nor fixed uncontrollable modes \citep{Hosoe1984StructuralCA,Murota_Book,zhang2019structural}}}. As a special case,  when each $p_i$ in $[A(p),B(p)]$ has a rank-one coefficient matrix, i.e.,
\begin{equation} \label{Linear Parameterization} A(p)=A_0+\sum\limits_{i=1}^m g_ip_ih_{1i}^{\intercal}, B(p)=B_0+\sum\limits_{i=1}^m g_ip_ih^{\intercal}_{2i}\end{equation}
where $A_0\in {\mathbb R}^{n\times n}$, $B_0\in {\mathbb R}^{n\times r}$,  $g_i,h_{1i}\in {\mathbb R}^n$, and $h_{2i}\in {\mathbb R}^r$, $\forall i$, the
structural controllability of $(A(p),B(p))$ can be characterized as follows.

\begin{lemma} (\citealt{Morse_1976, zhang2019structural}) \label{Theorem of Linear Parameterization}
Consider $(A(p),B(p))$ in (\ref{Linear Parameterization}). Define two transfer function matrices as $G_{zv}(\lambda)=[h_{11},\cdots,h_{1k}]^{\intercal}(\lambda I-A_0)^{-1}\![g_1,\cdots,g_k]$, $G_{zu}(\lambda)=[h_{11},\cdots,h_{1k}]^{\intercal}(\lambda I-A_0)^{-1}B_0+[h_{21},\cdots,h_{2k}]^{\intercal}$. The following statements are true:

  1) There is no parameter-dependent uncontrollable mode for the plant $(A(p),B(p))$, if and only if every cycle is input-reachable in ${\mathcal G}_{\rm aux}(G_{zv}(\lambda),G_{zu}(\lambda))$;

   2) There is no fixed uncontrollable mode for the plant $(A(p),B(p))$, if and only if for each $\lambda_0\in \sigma(A_0)$, ${\rm grank}[\lambda_0 I- A(p), B(p)]=n$, where ${\rm grank}(\cdot)$ is the maximum rank a matrix can achieve as a function of its indeterminates.
\end{lemma}

\section{NRCSs with SISO subsystems} \label{section_SISO}
This section derives conditions for system (\ref{sub_dynamic})--(\ref{sub_interaction}) to be structurally controllable when each subsystem is SISO, based on a constructive method. Since $r=1$,  let $c\doteq C$, $b\doteq B$ and $L=L_1=[-l_{ij}]$ for simplicity. Then (\ref{lump_pp_MIMO}) becomes
$A_{\rm sys}=I_N\otimes A-L\otimes bc$, and $B_{\rm sys}=\Delta \otimes b$.

\begin{theorem} \label{main_theorem}
Suppose $r=1$ and $|{\cal I}_u|<N$.\footnote{This assumption is made to avoid the trivial case where $|{\cal I}_u|=N$. If $|{\cal I}_u|=N$, the NRCS (\ref{sub_dynamic})--(\ref{sub_interaction}) is always structurally controllable provided $(A,b)$ is controllable (necessary for controllability), as $L=0$ makes the associated system controllable.}  Then, the NRCS (\ref{sub_dynamic})--(\ref{sub_interaction}) is structurally controllable, if and only if

1) $(A,b)$ is controllable and $(A,c)$ is observable;

2) $\bar {\cal G}_{\rm sys}$ is globally input-reachable.
\end{theorem}

The above result generalizes those of \cite{goldin2013weight} and \cite{Kazemi2018Structural} on structural controllability of networks of single-integrators running the consensus protocol. {Compared with \citet[Th. 2]{Commault2019Generic}, where all interaction weights are independent, Theorem \ref{main_theorem} has a simpler form without explicitly requiring the cacti condition therein. This is not surprising, as the zero row sum constraints of $L$ naturally induce a self-loop for each vertex of ${\cal G}_{\rm sys}$.}

{We are now giving a constructive and self-contained proof of Theorem \ref{main_theorem}, which is partially inspired by the techniques of \cite{menara2018structural}. Our proof enables a weight design procedure, as well as an extension to the case with heterogeneous subsystems (Corollary \ref{corollary_het2} of Section \ref{section_V}).} To this end, the following intermediate results are needed.

\begin{lemma} \citep{Modern_Control_Ogata} \label{BasicPoleAssign}
 Given $A\in {\mathbb R}^{n\times n}$, $b\in {\mathbb R}^{n}$, and $c\in {\mathbb R}^{1\times n}$, suppose that $(A,b)$ is controllable and $(A,c)$ is observable. Let $\Omega \subseteq {\mathbb C}$ be an arbitrary set of finite number of complex values.  Then,  there always exists $p\in {\mathbb R}$, such that $\sigma(A-pbc)\cap \Omega=\emptyset$.
\end{lemma}

\begin{lemma} \label{ControlNotZero}
 Given $A\in {\mathbb R}^{n\times n}$, $B\in {\mathbb R}^{n\times r}$, and $c_0\in {\mathbb R}^{1\times n}$, suppose that $(A,B)$ is controllable and $c_0\ne 0$. Then, $c_0(\lambda I-A)^{-1}B \not\equiv 0$.
\end{lemma}
\begin{proof}
We resort to the theory of {\emph{output controllability}} \citep[Sec. 9.6]{Modern_Control_Ogata}. If $(A,B)$ is controllable, then $(I_n, A, B)$ is output controllable. This requires that, the rows of $(\lambda I-A)^{-1}B$ are linearly independent in the field of complex numbers. That is, there cannot exist a nonzero $c_0 \in {\mathbb C}^{1\times n}$ such that $c_0(\lambda I-A)^{-1}B \equiv 0$.
\end{proof}

{\bf\emph{Proof of Theorem \ref{main_theorem}:}} (Necessity)  The necessity of Condition 1) follows directly from \citet[Th. 1]{Y_Zhang_2016} and \citet[Th. 4]{L.Wa2016Controllability}.  The proof for the necessity of Condition 2) is quite standard in structured system theory by leveraging the relationship between global input-reachability of $\bar {\cal G}_{\rm sys}$ and reducibility of $[L,\Delta]$ \citep[see][Sec. 4]{generic}. The details are omitted due to space considerations.

(Sufficiency: controllability of a tree) We use mathematical induction to prove the sufficiency.
First, without losing of generality, assume that there is a spanning tree $\cal T$ with the topological order $(1,...,N)$ in ${\cal G}_{\rm sys}$ and $\delta_1\ne 0$. Suppose $l_{ij}=0$ for $(j,i)\notin E(\cal T)$.  Let $A_k$ be the submatrix of $A_{\rm sys}$ formed by its first $kn$ rows and $kn$ columns, and $B_k=[b^{\intercal}, 0_{1\times (k-1)n}]^{\intercal}$. Consider $A_1=A$, $B_1=b$. It is obvious that $(A_1,B_1)$ is controllable. Now suppose that $(A_i,B_i)$ is controllable for $i=1,...,k$. Let $A_{k+1}$ be partitioned as
{\small\begin{equation}\label{matrix_partion}A_{k+1}=\left[
        \begin{array}{cccc}
          A_{11} & 0 & 0 & 0\\
          A_{21} & A_{22} & 0 & 0\\
          A_{31} & A_{32} & A_{33} & 0\\
          0 & l_{k+1}bc & 0 & A_{44}
        \end{array}
      \right]
\end{equation}}where $A_{11}\in {\mathbb R}^{n_1\times n_1}$, $A_{22}\in {\mathbb R}^{n\times n}$, $A_{33}\in {\mathbb R}^{n_2\times n_2}$,  $n_1+n+n_2=kn$ with $n_1$ and $n_2$ being divisible by $n$,  and $A_{44}=A-p_{k+1}bc$, with $p_{k+1}\in {\mathbb R}$ being the weight of the edge connecting vertex $k+1$ and its parent in $\cal T$. Note that the first three row and three column blocks of $A_{k+1}$ form $A_k$.  We will show that, by suitably choosing $p_{k+1}$, $[A_{k+1}-\lambda I,B_{k+1}]$ is of full row rank for each $\lambda\in {\mathbb C}$, which means that $(A_{k+1},B_{k+1})$ is controllable by the PBH test. To this end, consider the following two cases:

Case i) $n_1\ne 0$: Note that $[A_{k+1}-\lambda I, B_{k+1}]$ reads as
{\small\[\left[
  \begin{array}{ccccc}
    A_{11}-\lambda I & 0 & 0 & 0 & \bar b \\
    A_{21} & A_{22}-\lambda I & 0 & 0 & 0 \\
    A_{31} & A_{32} & A_{33}-\lambda I & 0 & 0 \\
    0 & p_{k+1}bc & 0 & A_{44}-\lambda I & 0
  \end{array}
\right],\]}where $\bar b= [b^{\intercal},0_{1\times (n_1-n)}]^{\intercal}$.
If $\lambda \notin \sigma(A_{44})$, as $(A_{k},B_k)$ is controllable, it can be directly validated that ${\rm rank} [A_{k+1}-\lambda I, B_{k+1}]=(k+1)n$. Consider $\lambda \in \sigma(A_{44})$. Recall $A_{44}=A-p_{k+1}bc$ and $(A,b)$ is controllable meanwhile $(A,c)$ is observable. From Lemma \ref{BasicPoleAssign}, there exists suitable $p_{k+1}$, such that $\sigma(A_{k})\cap \sigma(A_{44})=\emptyset$. Using the Schur complement \citep{Geor1993Graph}, when $\lambda\notin \sigma(A_k)$,  $[A_{k+1}-\lambda I, B_{k+1}]$ is of full row rank, if and only if
\begin{equation}\label{matrix_process} \begin{array}{l}
[A_{44}-\lambda I,0]-[0,p_{k+1}bc,0](A_{k}-\lambda I)^{-1}{\tiny\left[
                                                 \begin{array}{cc}
                                                   0 & \bar b \\
                                                   0 & 0 \\
                                                   0 & 0 \\
                                                 \end{array}
                                               \right]} \\
 =[A-p_{k+1}bc- \lambda I, p_{k+1}bc(A_{22}-\lambda I)^{-1}A_{21}(A_{11}-\lambda I)^{-1}\bar b]
 \end{array}\end{equation} is of full row rank.  Note that $c(A_{22}-\lambda I)^{-1}A_{21}(A_{11}-\lambda I)^{-1}\bar b$ is a scalar, and $A-p_{k+1}bc$ can be seen as a state feedback with feedback matrix $p_{k+1}c$. As $(A,b)$ is controllable,  the aforementioned condition is satisfied, if $c(A_{22}-\lambda I)^{-1}A_{21}(A_{11}-\lambda I)^{-1}\bar b \ne 0$ and $p_{k+1}\ne 0$. The first part of the latter condition is equivalent to
{\small\begin{equation}\label{not_zero_equal} [0,-c]\left[
          \begin{array}{cc}
            A_{11}-\lambda I & 0 \\
            A_{21} & A_{22}-\lambda I \\
          \end{array}
        \right]^{-1}\left[                      \begin{array}{c}
                        \bar b \\
                        0 \\
                      \end{array}
                    \right]
\ne 0.\end{equation}}By Lemma \ref{ControlNotZero}, noting that ${\tiny{\left(\left[
          \begin{array}{cc}
            A_{11} & 0 \\
            A_{21} & A_{22} \\
          \end{array}
        \right], \left[
                      \begin{array}{c}
                        \bar b \\
                        0 \\
                      \end{array}
                    \right]\right)}}$ is controllable (as $(A_i,B_i)$ is controllable for $i=1,...,k$), there exist only a finite number of complex values $\lambda$ such that  (\ref{not_zero_equal}) cannot hold. Let $\Omega_k=\{\lambda \in {\mathbb C}: \lambda \notin \sigma(A_{k}), c(A_{22}-\lambda I)^{-1}A_{21}(A_{11}-\lambda I)^{-1}\bar b=0\}$, then $\Omega_k$ is a finite set. Therefore, from Lemma \ref{BasicPoleAssign}, by suitably choosing $p_{k+1}\ne 0$, one can always ensure $(\sigma(A_{k})\cup \Omega_k)\cap \sigma(A-p_{k+1}bc)=\emptyset$, making $(A_{k+1},B_{k+1})$ controllable.

        Case ii) $n_1=0$: Following similar arguments, one can choose $p_{k+1}\ne 0$ to ensure $(\sigma(A_{k})\cup \Omega_k)\cap \sigma(A-p_{k+1}bc)=\emptyset$ with $\Omega_k\doteq\{\lambda \in {\mathbb C}: \lambda \notin \sigma(A_{k}), c(A_{22}-\lambda I)^{-1}b = 0\}$, so that $p_{k+1}$ makes $(A_{k+1},B_{k+1})$ controllable.

(Controllability of the network) If ${\mathcal G}_{\rm sys}$ can be~decomposed into more than one disjoint trees all rooted at the driving vertices, let the weights of edges between any two trees be zero.  Then, each tree itself corresponds to a controllable system. Thus, the whole system is controllable.~\hfill $\square$

Following the proof of Theorem \ref{main_theorem},  we provide a deterministic procedure to generate a set of interaction weights for an NRCS to be controllable with given SISO subsystems. For simplicity of description, assume that ${\cal G}_{\rm sys}$ can be spanned by a tree $\cal T$ with the topological order $(1,...,N)$, while $\delta_1\ne 0$. Let $p_{k}$ be the weight of the edge between vertex $k$ and its parent in $\cal T$, $k=2,...,N$, and let the weights of edges not in $\cal T$ be zero.  For $k=1,...,N-1$,  $p_{k+1}$ can be recursively constructed in the following way, ~where $A_k$, $A_{11},...,A_{22}$ are defined in the proof of Theorem~\ref{main_theorem}:

1) partition $A_{k+1}$ according to (\ref{matrix_partion});

2) if $A_{11}$ is not empty, let $\Omega_k=\{\lambda \in {\mathbb C}: \lambda \notin \sigma(A_{k}), c(\\A_{22}-\lambda I)^{-1}A_{21}(A_{11}-\lambda I)^{-1}\bar b=0\}$, where $\bar b= [b^{\intercal},0_{1\times (n_1-n)}]^{\intercal}$, otherwise let $\Omega_k=\{\lambda \in {\mathbb C}: \lambda \notin \sigma(A_{k}), c(A_{22}-\lambda I)^{-1}b = 0\}$. Determine an $p_{k+1}\in {\mathbb R}$, such that $\sigma(A_{k})\cap \sigma(A-p_{k+1}bc)=\emptyset$ and $\Omega_k \cap \sigma(A-p_{k+1}bc)=\emptyset$.

Step 2) can be implemented by leveraging standard pole-assignment techniques for SISO systems \citep[Chap. 10]{Modern_Control_Ogata}. A corollary can be obtained from this procedure, which is vital in the proof of Theorem \ref{MIMO_equal} (see Section \ref{section_single}).

\begin{corollary} \label{no_repeated_L} Let $L$ be a Laplacian matrix of a graph $\cal G$ with $N$ vertices $\{1,...,N\}$. Suppose that $\cal G$ has a spanning tree rooted at vertex $1$.  Then, there exists a set of weights for $\cal G$ such that the associated $(L, {\bf e}^{N}_1)$ is controllable while $L$ has no repeated eigenvalues.
\end{corollary}

\begin{proof} By setting $A=0\in{\mathbb R}$, $b=c=1$, the previous procedure provides a way to construct such an $L$.
\end{proof}
{
\begin{remark} Example 2 of \cite{menara2018structural} has discussed the controllability of discrete-time consensus networks and shown that it remains a generic property. A discrete-time consensus network of single-integrators has a state-transition matrix whose every row sum equals one (zero in the continuous-time case).  If we use $I_N-L$ to denote such a state-transition matrix, where $L$ is the Laplacian matrix of the associated $N$-vertex network, then by setting $A=1$ Theorem \ref{main_theorem} indicates that global-input reachability is still necessary and sufficient for structural controllability of the discrete-time network (note that the PBH test for discrete-time systems is the same as that for continuous-time ones). This is consistent with \citet[Th. 3.3]{menara2018structural} in providing necessary conditions for the structural controllability of undirected consensus networks.
\end{remark}}

\section{NRCSs with MIMO subsystems via equally \\ weighted channels} \label{section_single} 
In this section, we generalize the results in the above section to the case with MIMO subsystems via equally weighted channels, i.e., the interaction fashion (b). For notation simplicity, let $L_1=\cdots=L_r=L=[-l_{ij}]$, and rewrite (\ref{lump_pp_MIMO}) as $A_{\rm sys}=I\otimes A-  L\otimes BC, B_{\rm sys}=\Delta \otimes B$.
We will give some testable conditions for structural controllability. Our approach is based on the mode peculiarity of $A_{\rm sys}$.

{\begin{proposition}\label{new_proposition_nopara}  Consider the NRCS (\ref{sub_dynamic})--(\ref{sub_interaction}) with $L_1=\cdots=L_r$. Suppose that ${{\cal G}}_{\rm sys}$ has a spanning tree rooted at one driving vertex. Then, there is no parameter-dependent uncontrollable mode for this system. In other words, under that condition, the NRCS is structurally controllable, if and only if for each $\lambda_i\in \sigma(A)$, the following matrix has full row generic rank.
\begin{equation}\label{fullrank} [\lambda_iI_{nN}-I_{N}\otimes A+ L\otimes BC, \Delta\otimes B].\end{equation}
\end{proposition}}
\begin{proof} See the appendix. \end{proof}

{Note that the existence of a spanning tree is necessary for structural controllability in the single-input case (i.e., $|{\cal I}_u|=1$). The above proposition indicates that, under this condition, all parameter-dependent modes are generically controllable. As such, verifying structural controllability is transformed into the problem of generic rank verifications at some fixed modes. One could resort to the matrix net techniques in \cite{Anderson_1982} to check the full row generic rank of (\ref{fullrank}) at each $\lambda_i\in \sigma(A)$. It seems nontrivial to extend this proposition to the case with multiple inputs. Nevertheless, based on Proposition \ref{new_proposition_nopara}, the following theorem gives a sufficient condition for structural controllability, which avoids checking the generic rank of (\ref{fullrank}) at the global system level.}

\begin{theorem} \label{MIMO_equal} Given the NRCS (\ref{sub_dynamic})--(\ref{sub_interaction}) with $L_1=\cdots=L_r$, suppose that $\bigcap \nolimits_{l\in {\mathbb R}}\sigma(A+lBC)=\emptyset$. Then, this system is structurally controllable, if and only if $\bar {\cal G}_{\rm sys}$ is globally input-reachable.
\end{theorem}

\begin{proof} (Necessity) The necessity follows similar arguments to those of the proof for Theorem \ref{main_theorem}. 

(Sufficiency) First, assume that ${\cal G}_{\rm sys}$ has a spanning tree with the topological order $(1,..., N)$ while $\delta_1\ne 0$. Denote this tree by $\cal T$. From Proposition \ref{new_proposition_nopara}, system (\ref{sub_dynamic})--(\ref{sub_interaction}) has no parameter-dependent uncontrollable modes. To show this system has no fixed uncontrollable modes, it suffices to show that (\ref{fullrank}) has full row generic rank for each eigenvalue of $A$.  Let the weight of the edge connecting vertex $i$ and its parent by $p_i$, $i=2,...,N$, and $p_1\equiv 0$, while weights of edges not in $E({\cal T})$ be zero. Then, the $i$th diagonal block of $A_{\rm sys}$ can be written as $A-p_iBC$. Because $\bigcap \nolimits_{l\in {\mathbb R}}\sigma(A+lBC)=\emptyset$, we have ${\rm grank}[\lambda_jI-A+p_iBC]=n$, and ${\rm rank}[\lambda_jI-A,B]=n$  for each $\lambda_j\in \sigma(A)$ (as that condition requires $(A,B)$ to be controllable). After some row and column permutations, $[\lambda_j I- A_{\rm sys}, B_{\rm sys}]$ has a lower block triangular form, whose $1$st diagonal block, being $[\lambda_j I-A, B]$, and $2$nd to $N$th diagonal blocks, being $\lambda_jI-A+p_iBC$, are all of full row generic rank. Therefore, ${\rm grank}[\lambda_j I- A_{\rm sys}, B_{\rm sys}]=nN$.

If ${\cal G}_{\rm sys}$ can be decomposed into more than one disjoint trees all rooted at the driving vertices, let the edges connecting these trees have weight zero. It turns out each tree itself corresponds to a structurally controllable system. \! Thereby, the whole system is structurally controllable.~\end{proof}

{ From} \citet[Lem. 4.1]{Anderson_1982}, $\bigcap \nolimits_{l\in {\mathbb R}}\sigma(A+lBC)=\emptyset$, if and only if the following matrix has full rank for each $\lambda_0\in \sigma(A)$
\[{\footnotesize\left[
  \begin{array}{ccccc}
    A-\lambda_0I & BC & 0 & \cdots & 0 \\
    0 & A-\lambda_0I & BC & \cdots & 0 \\
    \vdots &  &  & \ddots & \vdots  \\
    0 & 0 & 0 & \cdots & BC \\
    0 & 0 & 0 & \cdots & A-\lambda_0I \\
  \end{array}
\right]\in {\mathbb R}^{(\bar l+1)n\times (\bar l+1)n}}
,\]where $\bar l\doteq {\rm rank}BC$. {When such condition is not satisfied, suppose $\Psi\doteq \bigcap \nolimits_{l\in {\mathbb R}}\sigma(A+lBC)$. From Proposition \ref{new_proposition_nopara} and Theorem \ref{MIMO_equal}, provided that there is a spanning tree rooted at one driving vertex in ${\cal G}_{\rm sys}$, to verify structural controllability one only needs to check the generic rank of (\ref{fullrank}) at each $\lambda_i\in \Psi\subseteq \sigma(A)$.}

\begin{remark} It should be noted that the condition  $\bigcap \nolimits_{l\in {\mathbb R}}\sigma(A+lBC)=\emptyset$ has also been proposed in \cite{xue2018modal}. However, different from \cite{xue2018modal} where the interaction weights are fixed and form a diagonalizable matrix, in this paper, we study controllability in the generic sense, where the weights are indeterminates without the diagonalization assumption. Our results directly link to the network topology's graphical properties, rather than the spectrum of the matrix formed by the interaction weights.
\end{remark}

\section{NRCSs with MIMO Subsystems via differently weighted channels} \label{MIMO_section}
In this section, we consider NRCSs with MIMO subsystems via differently weighted channels, i.e., the interaction fashion (c). Note that characterizing modes of $A_{\rm sys}$ is vital to Theorem \ref{MIMO_equal}. However, it~seems unlikely to implement similar analysis if $L_1,...,L_r$ are nonidentical. Our derivations are based on Lemma~\ref{Theorem of Linear Parameterization}.

 For Lemma \ref{Theorem of Linear Parameterization} to be used, we need to linearly parameterize $L_i$ as in (\ref{Linear Parameterization}).
 To this end, define the incidence matrix $K_I$ of ${\cal G}_{\rm sys}$ as the $|{\cal E}_{xx}|\times |{\cal V}_{x}|$ matrix such that for the $k$th edge $e_k=(l,j)$, $[K_I]_{kl}=1$ and $[K_I]_{kj}=-1$. Afterwards, define a $|{\cal V}_{x}|\times |{\cal E}_{xx}|$ matrix $K$ as $K_{jk}=1$ if ${[{K_I}]_{kj}} =  - 1$, and otherwise $K_{jk}=0$.
It turns out $L_i=-K\Lambda_iK_I$, where $\Lambda_i$ is a diagonal matrix whose $k$th diagonal equals the weight of $e_k$ associated with $L_i$. We then have
\begin{equation} \label{linear_pa}
\begin{split}
[A_{\rm sys},B_{\rm sys}]=[I\otimes A, \Delta \otimes B]+[K\otimes b_1,...,K\otimes b_r]\\
{\bf diag}\{\Lambda_1,...,\Lambda_r\}[[K_I^{\intercal}\otimes c_1^{\intercal},...,K_I^{\intercal}\otimes c_r^{\intercal}]^{\intercal},0].
\end{split}
\end{equation}

To extend the SISO case to the MIMO one, we  draw the following notion from decentralized stabilization. 

\begin{definition} \citep{Fixed_Mode} \label{fix_mode} Given~$A\in {\mathbb R}^{n\times n}$, $B\in {\mathbb R}^{n\times r}$ and $C\in {\mathbb R}^{r\times n}$, let ${\cal K} \subseteq {\mathbb R}^{r\times r}$ be the set of all $r\times r$ diagonal matrices. Then $(A,B,C)$ is said to have no fixed mode with respect to $\cal K$, if $\bigcap \nolimits_{K\in {\cal K}} \sigma(A+BKC)\!\!=\!\!\emptyset$.
\end{definition}

To proceed with our derivations, we need the following result, whose proof can be found in  \cite{zhang2019Structural_arxiv}.
\begin{lemma} \label{lemma_equal_graph}
Given matrices $H\in {\mathbb M}^{k\times n}$, $P \in {\mathbb M}^{k\times m}$, $G\in {\mathbb M}^{n\times k}$, and a diagonal matrix $\Lambda \in {\mathbb M}^{n\times n}$ whose diagonal entries are free parameters, suppose the following condition holds:  $[GH]_{ij}\ne 0$ (resp. $[GP]_{ij}\ne 0$) whenever there exists one $l\in \{1,...,k\}$ such that $G_{il}\ne 0$ and $H_{li}\ne 0$ (resp. $P_{li}\ne 0$). Then, every cycle is input-reachable in ${\cal G}_{\rm aux}(GH,GP)$, if and only if such property holds in ${\cal G}_{\rm aux}(H\Lambda G,P)$.
\end{lemma}

\begin{theorem}\label{MIMO}
For the NRCS (\ref{sub_dynamic})--(\ref{sub_interaction}) with MIMO subsystems via differently weighted channels, suppose that $c_i\ne 0$ for $i=1,...,r$.  The following statements are true:

{1) If $(A,B)$ is controllable and $\bar {\cal G}_{\rm sys}$ is globally input-reachable, then there is no parameter-dependent uncontrollable mode;}

2) Suppose that $(A,B,C)$ has no fixed mode w.r.t. $\cal K$. Then, the networked system is structurally controllable, if and only if $\bar {\cal G}_{\rm sys}$ is globally input-reachable.
\end{theorem}
\begin{proof}  We first prove Statement 1). When Lemma \ref{Theorem of Linear Parameterization} is used on (\ref{linear_pa}), direct algebraic manipulations show that the associated transfer function matrices are
{\footnotesize \[\begin{aligned}
&{G_{{{zv}}}}(\lambda ) = {\bf col}\{K_I\otimes c_i|_{i=1}^r\}{(\lambda I - I \otimes A)^{ - 1}}[K \otimes b_1,...,K \otimes b_r]\\
&=\left[
    \begin{array}{ccc}
      {{K_I}K \otimes {c_1}{{(\lambda I - A)}^{- 1}}{b_1}} & \cdots & {{K_I}K \otimes {c_1}{{(\lambda I - A)}^{- 1}}{b_r}} \\
      \vdots & \cdots & \vdots \\
      {{K_I}K \otimes {c_r}{{(\lambda I - A)}^{- 1}}{b_1}} & \cdots & {{K_I}K \otimes {c_r}{{(\lambda I - A)}^{- 1}}{b_r}} \\
    \end{array}
  \right] \\
&{G_{{{zu}}}}(\lambda) = {\bf col}\{K_I\otimes c_i|_{i=1}^r\}[\Delta\otimes ({(\lambda I - A)^{ - 1}}B)]\\
&=\left[
    \begin{array}{c}
      (K_I\Delta)\otimes (c_1(\lambda I-A)^{-1}B) \\
      \vdots \\
      (K_I\Delta)\otimes (c_r(\lambda I-A)^{-1}B)
    \end{array}
  \right].
  \end{aligned}\]
  }Partition $G_{zv}(\lambda)$ into $r\times r$ blocks, where the $(i,j)$th block is $(K_IK)\otimes(c_i(\lambda I- A)^{-1}b_j)$. From Lemma \ref{ControlNotZero}, there is at least one nonzero block in each row block of $G_{zv}(\lambda)$. Suppose that the $(i,\sigma(i))$th block is nonzero, $i=1,...,r$, $\sigma(i)\in\{1,...,r\}$. Let $\bar {G}_{zv}$ be the matrix with the same dimensions and partitions as ${G}_{zv}(\lambda)$ by setting its $(i,\sigma(i))|_{i=1}^r$th blocks to be $K_IK$ and the rest zero. Similarly, partition $G_{zu}(\lambda)$ into $r\times 1$ blocks, where the $i$th row block is $(K_I\Delta)\otimes(c_i(\lambda I-A)^{-1}B)$. Again from Lemma \ref{ControlNotZero}, each of its row blocks is nonzero. Let $\bar G_{zu}=[(K_I\Delta)^{\intercal},...,(K_I\Delta)^{\intercal}]^{\intercal}$, and define $\Delta_U\doteq [\Delta^{\intercal},...,\Delta^{\intercal}]^{\intercal}$. It is now easy to see that, if  {\small ${\cal G}_{\rm aux}(\bar G_{zv}, \bar G_{zu})$} is globally input reachable, then {\small ${\cal G}_{\rm aux}(G_{zv}(\lambda),G_{zu}(\lambda))$} will be.

  Define matrices $G\doteq {\bf diag}\{K_I|_{i=1}^r\}$, $P\doteq \Delta_U$, and let $H$ be such that $GH=\bar G_{zv}$. Then, $(\bar G_{zv}, \bar G_{zu})$ can be written as $(GH,GP)$. Let $\Lambda$ be a diagonal matrix whose diagonal entries are free parameters, and $L$ a Laplacian matrix associated with ${\cal G}_{\rm sys}$.  Using Lemma \ref{lemma_equal_graph} on $(H,P,G,\Lambda)$,  we obtain the matrix
$\left[ {{L_U},{\Delta _U}} \right],$ where $L_U$ is a matrix with $r\times r$ blocks with the $(i,\sigma(i))|_{i=1}^r$th block being $L$ and each of the rest being the $N\times N$ zero matrix.

Now assume that there is a spanning tree in ${\cal G}_{\rm sys}$ with the topological order $(1,...,N)$, denoted by $\cal T$, while $\delta_1\ne 0$.  Denote the parent of vertex $k$ by ${\rm Par}(k)$. Then ${\rm Par}(k)\in\{u_1,1,...,k-1\}$.  It can be observed that the $(N(i-1)+j, N(\sigma(i)-1)+{\rm Par}(j))$th entry of $L_U$ is nonzero, and the $(N(i-1)+1, 1)$th entry of $\Delta_U$ is also nonzero, for each $1\le i\le r, 1\le j \le N$. Denote the vertex of ${\cal G}_{\rm aux}(L_U,\Delta_U)$ associated with the $(N(i-1)+j)$th row of $L_U$ by the pair $\{i,j\}$. Based on these observations, vertex $\{i,j\}$ has an ingoing edge from vertex $\{\sigma(i), {\rm Par}(j)\}$, and vertex $\{i,1\}$ is always input-reachable, for $1\le i\le r, 2\le j \le N$. Consequently, for each vertex $\{i,j\}$, there is a path from {\small $\{\underbrace{{\sigma(\cdots(\sigma(i))\cdots)}}_{(j-1)\quad \sigma(\cdot)}, \underbrace{{{\rm Par}(\cdots({\rm Par}(j))\cdots)}}_{(j-1) \quad {\rm Par}(\cdot)}\}$} to it. From the fact {\small $\underbrace{{{\rm Par}(\cdots({\rm Par}(j))\cdots)}}_{(j-1)\quad {\rm Par}(\cdot)}\}=1$ }, it concludes that every vertex of ${\cal G}_{\rm aux}(L_U,\Delta_U)$ is input-reachable. By Lemma \ref{lemma_equal_graph}, there is no input-unreachable cycle in {\small${\cal G}_{\rm aux}(G_{zv}(\lambda),G_{zu}(\lambda))$}. By Lemma \ref{Theorem of Linear Parameterization}, Statement 1) follows immediately. The case that ${\cal G}_{\rm sys}$ can be decomposed into more than one trees rooted at the driving vertices  follows similar arguments.

We then prove Statement 2). The necessity of the global input-reachability follows similar arguments to the proof of Theorem \ref{main_theorem},  thus omitted here. For sufficiency,  because of Statement 1) we only need to prove that Condition 2) of Lemma \ref{Theorem of Linear Parameterization} is satisfied. Owing to the absence of subsystem fixed modes, this can be proved by following similar arguments to the proof for sufficiency of Theorem \ref{MIMO_equal}, where we just need to replace the corresponding $A-p_iBC$ with $A+B{\bf diag}\{l^{[k]}_{ii}|_{k=1}^r\}C$ for $i=1,...,N$. Details are omitted due to their similarities.
\end{proof}

 {Statement 1) of Theorem \ref{MIMO} indicates that, under the necessary conditions for structural controllability (namely, the controllability of $(A,B)$ [\citealt{Y_Zhang_2016}] and global input-reachability of ${\bar {\cal G}}_{\rm sys}$), all parameter-dependent modes are generically controllable.}  Statement 2) of Theorem \ref{MIMO} means that, with the absence of subsystem fixed modes, global input-reachability is sufficient for structural controllability. It is easy to see that Theorem \ref{main_theorem} is a special case of Theorem \ref{MIMO}. {Additionally, when $(A,B,C)$ has some fixed modes w.r.t. $\cal K$, provided that the necessary conditions in Statement 1) hold, a testable procedure for structural controllability of the NRCS is to verify the row generic rank of the matrix similar to (\ref{fullrank}) at each of these fixed modes.\footnote{There are many criteria to verify whether $(A,B,C)$ has fixed modes w.r.t. $\cal K$, including the algebraic criterion and the matroid based criterion \citep[Ch. 6.5]{Murota_Book}. The latter enables polynomial time complexity.}  Since each indeterminate has a rank-one coefficient matrix therein, this can be done in polynomial time via the tool of matroid intersection \citep{Murota_Book}. { However, in such case, the necessary and sufficient conditions seem to depend on $(A,B,C)$ and $\bar {\cal G}_{\rm sys}$ in a complicatedly coupled way, which is hard to be presented in simple graph-theoretic forms \citep{Murota_Book}.}}

\begin{remark}
It can be seen that if $\bigcap \nolimits_{l\in {\mathbb R}}\sigma(A+lBC)=\emptyset$, then $(A,B,C)$ has no fixed mode w.r.t. $\cal K$, but the inverse is not necessarily true. This means the condition of Theorem \ref{MIMO} is less restrictive than that of Theorem \ref{MIMO_equal}, which is reasonable, as allowing heterogeneous interaction weights permits more freedom for weight assignment. {{However, Theorem \ref{MIMO_equal} cannot be directly obtained from Theorem \ref{MIMO}. This is because the set of weights satisfying $L_1=\cdots=L_r$  has zero Lebesgue measure in ${\mathbb R}^{|{\cal E}_{xx}|r}$ (the parameter space for the differently weighting~case). The essential differences between interaction fashions (b) and (c) lie in two aspects. On the one hand, the former fashion readily allows an analytical expression on how the global spectrum (modes of $A_{\rm sys}$) and eigenspaces depend on $(A,B,C)$ and the spectrum of the matrix $L$ formed by the interaction weights (see the proof of Proposition \ref{new_proposition_nopara}), while a similar analytical expression is unlikely to exist for the latter fashion. This fact means that the interaction weights may have affected the corresponding global spectra in different ways. On the other hand, each indeterminate in $A_{\rm sys}$ of fashion (b) has a coefficient matrix with a rank larger than one if ${\rm rank} BC>1$, while this rank equals one in fashion (c). Unlike the rank-one case, verifying the existing criteria for structural controllability in the high-rank case generally requires exponential computational burden in the dimension of the problem \citep{Anderson_1982,zhang2019structural}. This makes techniques towards these two fashions different. }}
\end{remark}

\section{Extensions with subsystem heterogeneities and application examples} \label{section_V}

This section will discuss extending results in the previous sections to the cases with subsystem heterogeneities and show their direct applications on some typical practical systems.  Specifically, all these examples involve subsystems with heterogeneous parameters.

When modeling real-world networked systems, it is often the case that subsystems obey the same physical laws, thus parameterized similarly, but possibly with different values of their {\emph{elementary parameters}}. Here, the elementary parameters refer to parameters that directly describe the movements of subsystems (for example, in the mass-spring-damper system in Fig. \ref{damping_vehicle}, the mass $m_i$,  the constants of the spring $k_i$, and of the damper $\mu_i$ could be seen as elementary parameters).
At first sight, the inevitable subsystem heterogeneities caused by variants of subsystem elementary parameters might prevent our analysis from being applicable. However, our analysis and most results in the former sections are indeed applicable under certain of these heterogeneities. We will provide two of such cases.

The first case is that one could decouple the `heterogeneous part' from subsystem dynamics and put it into the interaction weights. If the structures of the associated Laplacian matrices are preserved after such operation, then most results in Sections \ref{section_SISO}--\ref{MIMO_section} could still be valid.

{\begin{example}[{Connected tank system}]Consider the \\~ liquid-level system with $N$ {connected tanks} shown in
\citet[Fig. 4.2]{Modern_Control_Ogata}. Assuming small variations of the variables from the steady-state values, the dynamics of the $i$th tank is ${h_i} - {h_{i + 1}} \!=\! {q_i}{R_i},{C_i}{\dot h_i} \!=\! {q_{i - 1}} - {q_i}$, where $h_i$ is the head of the liquid level, $q_i$ is the outflow rate,  $C_i$ and $R_i$ are the capacitance of the tank and the resistance of liquid flow in the pipe, respectively, for $i\in\{1,...,N\}$. Here, $q_0$ should be regarded as the input rate, and $h_{N+1}\!\equiv\!0$. Rewrite the previous equation as $${{\dot h}_i}= l_{i,i-1}({h_{i - 1}} - {h_i})+ l_{i,i+1}({h_{i + 1}} - {h_i}),$$ with $l_{i,i-1}\!\doteq\!1/(C_iR_{i - 1})$ and $l_{i,i+1}\!\doteq\!1/({C_i}{R_i})$. Regarding $\{C_i,R_i\}|_{i=1}^N$ as independent indeterminates, there is no algebraic dependence among the nonzero off-diagonal entries of the associated matrix $L\!=\![-l_{ij}]$. By Theorem \ref{main_theorem}, since the considered {connected tank} system has a chain structure, we conclude that it is structurally controllable.
\end{example}}

\begin{example}[The motivating example continuing]Let us revisit the mass-spring-damper system in Fig. \ref{damping_vehicle}. Regarding $\{m_i,k_i,\mu_i\}|_{i=1}^N$ as independent indeterminates, there is no algebraic dependence among the weights $\{l^{[1]}_{ij}\}$ and $\{l^{[2]}_{ij}\}$.  Moreover, it can be validated that the associated $(A,B,C)$ has no fixed mode. As the network topology is a chain, this system is structurally controllable by driving arbitrarily one mass from Theorem \ref{MIMO}.
\end{example}

The second case is that, the subsystem heterogeneities arising from variants in values of elementary parameters could be expressed by $A+\delta A_i$. Here $\delta A_i$ is a structured matrix, $i\in \{1,...,N\}$, and $\delta A_1,...,\delta A_N$ have the same structure, denoted by $\delta A$,  whereas their nonzero entries could take values independently (both within each $\delta A_i$ and between two different $\delta A_i$ and $\delta A_j$). For brevity, we only focus on the SISO subsystem case. In this regard, rewrite the $i$th subsystem dynamics (\ref{sub_dynamic}) as 
\begin{equation}\label{sub_heteroge}
\dot x_i(t)=(A+\delta A_i)x_i(t)+bv_i(t).
\end{equation}

\begin{corollary} \label{corollary_het} Consider the NRCS described by (\ref{sub_heteroge}) and (\ref{sub_interaction}). This system is structurally controllable, if 1) $(A+\delta A,b)$ is structurally controllable and $(A+\delta A,c)$ is structurally observable\footnote{This means that there exists one numerical realization of $\delta A$, denoted by $\overline {\delta A}$, such that $(A+\overline {\delta A}, b)$ is controllable and $(A+\overline {\delta A}, c)$ is observable.}; 2) $\bar {\cal G}_{\rm sys}$ is globally input-reachable.
\end{corollary}

\begin{proof} If 1) and 2) are satisfied, choose one numerical realization of $\delta A$, denoted by $\overline {\delta A}$, such that $(A+\overline {\delta A}, b)$ is controllable and $(A+\overline {\delta A}, c)$ is {observable}. Let $\delta A_i$ for each subsystem take the same value as ${\overline {\delta A}}$, $i\in \{1,...,N\}$. From Theorem \ref{main_theorem} and because of 2), the resulting system is structurally controllable.
\end{proof}

\begin{example}[Power network]Consider a power network consisting of $N$ generators. The dynamics of the $i$th generator around its equilibrium state could be described by the linearized Swing equation \citep{P.Fa2014Controllability}: ${m_i}{{\ddot \theta }_i} + {d_i}{{\dot \theta }_i} =  - \sum\nolimits_{j=1}^N {{k_{ij}}{\rm{(}}{\theta _i}{\rm{ - }}{\theta _j}{\rm{) + }}{P_{i}}}$, where $\theta_i$ is the phrase angle, $m_i$ and $d_i$ are respectively the inertia and damping coefficients, and $P_{i}$ is the input power, $i\in\{1,...,N\}$. $k_{ij}$ is the susceptance of the power line from the $j$th generator to the $i$th one. Rewrite the previous equation as 
{\small\begin{equation}\label{sub_power2} \begin{array}{l}
\left[ {\begin{array}{*{20}{c}}
{{{\dot \theta }_i}}\\
{{{\ddot \theta }_i}}
\end{array}} \right] = \underbrace {\left[ {\begin{array}{*{20}{c}}
0&1\\
0& -d_i/m_i
\end{array}} \right]}_{A + \delta A_i}\left[ {\begin{array}{*{20}{c}}
{{\theta _i}}\\
{{{\dot \theta }_i}}
\end{array}} \right]+\\
{\kern 1pt} \sum\limits_{j = 1,...,N} {\underbrace {\frac{{{k_{ij}}}}{{{m_i}}}}_{{l_{ij}}}} \underbrace {\left[ {\begin{array}{*{20}{c}}
0\\
1
\end{array}} \right]}_b\underbrace {\left[ {\begin{array}{*{20}{c}}
1&0
\end{array}} \right]}_c\left[ {\begin{array}{*{20}{c}}
{{\theta _j} - {\theta _i}}\\
{{{\dot \theta }_j} - {{\dot \theta }_i}}
\end{array}} \right] + \underbrace {\left[ {\begin{array}{*{20}{c}}
0\\
1
\end{array}} \right]}_b\frac{{{P_i}}}{{{m_i}}}.
\end{array}\end{equation}}Suppose that $\{m_i,d_i\}|_{i=1}^N$  are mutually independent. Then, {{${-d_i}/{m_i}$}} and ${k_{ij}}/{m_i}$ can be seen as indeterminates representing subsystem heterogeneities and the interaction weights, respectively. The considered power network model can be described by (\ref{sub_heteroge}) and (\ref{sub_interaction}), which is an NRCS with SISO subsystems. It can be validated that, $(A+\delta_i A, b)$ is structurally controllable and $(A+\delta_i A, c)$ is structurally observable. By Corollary \ref{corollary_het},  provided that there exists a path (consisting of power lines) from one input to each generator in the power system, this system is structurally controllable.
\end{example}

{Finally, suppose that subsystem parameters $(A,B,C)$ are heterogeneous. Specifically, (\ref{sub_dynamic})--(\ref{sub_interaction}) becomes
 \begin{equation}\label{sub_heteroge2}
\dot x_i(t)=A_{[i]}x_i(t)+b_{[i]}\big\{\delta_iu_i(t)+\sum \limits_{j=1,j\ne i}^N  l_{ij}(c_{[j]}x_j(t)-c_{[i]}x_i(t))\big\},
\end{equation}with known $A_{[i]}\in {\mathbb R}^{n_i\times n_i}$, $b_{[i]},c_{[i]}^{\intercal}\in {\mathbb R}^{n_i}$, and $n_{i}$ for different $i$ not needing to be identical.

\begin{corollary} \label{corollary_het2} Consider the NRCS with heterogeneous SISO subsystems described by (\ref{sub_heteroge2}). This system is structurally controllable, if 1) $(A_{[i]},b_{[i]})$ is controllable and $(A_{[i]},c_{[i]})$ is observable, $\forall i$; 2) $\bar {\cal G}_{\rm sys}$ is globally input-reachable.
\end{corollary}

\begin{proof} The proof is similar to that for sufficiency part of Theorem \ref{main_theorem}. In each of the induction in that proof, we just need to replace $A_{44}$ with  $A_{44}=A_{[k+1]}-p_{k+1}b_{[k]}c_{[k]}$, and then (\ref{matrix_process}) becomes
$[A_{[k+1]}-p_{k+1}b_{[k+1]}c_{[k+1]}- \lambda I, p_{k+1}b_{[k+1]}\\~ c_{[\overline{k+1}]}(A_{22}-\lambda I)^{-1}A_{21}(A_{11}-\lambda I)^{-1}\bar b],$
where $\overline{k+1}$ is the parent of vertex $k+1$ in the associated spanning tree. It is readily to see that the remaining statement of that proof is still true. Details are omitted.
\end{proof} }
\vspace*{-0.1em}
\section{Conclusions} \label{section_conl}
This paper studies the structural controllability of NRCSs in which subsystems are of identical and fixed high-order linear dynamics. Three types of subsystem interaction fashions are considered. {We have shown that, the NRCS has no parameter-dependent uncontrollable modes under some necessary connectivity conditions.} Some necessary and/or sufficient conditions are given for structural controllability depending on the subsystem dynamics and the network topologies in a decoupled form. Extensions to handle certain subsystem heterogeneities are also provided with various practical examples. {Our results are among the recent attempts \citep{Liu2019AGC,Commault2019Generic,menara2018structural} to give simple graph-theoretic conditions for structural controllability of networks with parameter dependencies, with an emphasis on subsystem dynamics and more complicated interaction fashions.} 

\vspace*{-0.1em}
\section*{Appendix: Proof of Proposition 1}
\renewcommand{\theequation}{\arabic{equation}}

{\begin{lemma} \label{notfixedratio}Consider $A_{\rm sys}$ with $L_1=\cdots=L_r=L$. Suppose $\cal G_{\rm sys}$ has a spanning tree. Let $p=(p_1,...,p_{{ \bar r}})$ with $p_i$ being an indetermiante denoting the weight of the $i$th edge in ${\cal E}_{xx}$, ${ \bar r}\doteq |{\cal E}_{xx}|$, and $L$ be the Laplacian matrix of $\cal G_{\rm sys}$. Let $\det(sI-A+p_iBC)=\alpha_0(s)\beta_0(s,p_i)$ with $\alpha_0(s)$ being an $s$-factor and $\beta_0(s,p_i)$ being an $(s,p)$-factor. Suppose that $\det(sI-A_{\rm sys})$ can be decomposed as $\det(sI-A_{\rm sys})=\alpha(s)\beta(s,p)$, where $\alpha(s)$ is an $s$-factor and $\beta(s,p)$ is an $(s,p)$-factor. Then, $\deg\beta(s,p)=(N-1)\deg\beta_0(s,p_i)$. That is, $A_{\rm sys}$ has $(N-1)\deg\beta_0(s,p_i)$ parameter-dependent modes (counting multiplicities).
\end{lemma}
\begin{proof}
Without losing generality, let $\cal T$ be the spanning tree of $\cal G_{\rm sys}$ with the topological order $(1,...,N)$ and assume that the edge connecting vertex $i$ and its parent has weight $p_i$, $i=2,...,N$. For each $p\in {\mathbb R}^{\bar r}$, suppose that $L$ has $l(p)\le N$ distinct eigenvalues, denoted by $\lambda_1,...,\lambda_{l(p)}$, and assume that each $\lambda_i$ has an algebraic multiplicity $\hat r_i$ (the dependence of $\lambda_i$ and $\hat r_i$ on $p$ has been dropped). Without losing any generality, assume that $\lambda_1=0$ (obviously $\hat r_1\ge 1$). Then, there exists an $N\times N$ invertible matrix $Q$ such that $Q^{-1}LQ=\Gamma$, where $\Gamma={\bf {diag}}\{\Gamma_i|_{i=1}^{l(p)}\}$, and $\Gamma_i$ is the block diagonal matrix consisting of all Jordan blocks of $L$ associated with $\lambda_i$. Notice that each Jordan block is such that the only (possible) non-zero entries are on the diagonal (being $\lambda_i$) and the superdiagonal (being $1$), and $\sum \nolimits_{i=1}^{l(p)} \hat r_i=N$.  Because $(Q^{-1}\otimes I_n)A_{\rm sys}(Q\otimes I_n)=I_N\otimes A-\Gamma\otimes BC$, we obtain $\det (sI-A_{\rm sys})=\prod \nolimits_{i=1}^{l(p)} [\det (sI-A+\lambda_i BC)]^{\hat r_i}=[\det(sI-A)]^{\hat r_1}\alpha_0^{N-\hat r_1}(s)\prod \nolimits_{i=2}^{l(p)}[\beta_0(s,\lambda_i)]^{\hat r_i}$ for any $p\in {\mathbb R}^{{ \bar r}}$. Note that $\lambda_2,...,\lambda_{l(p)}$ may not be independent. Hence, \begin{equation}\label{ineq2}\deg \alpha(s)\ge n\hat r_1+(N-\hat r_1)\deg \alpha_0(s)\ge n+(N-1)\deg \alpha_0(s).\end{equation}

On the other hand, define the Laplician matrix in~(\ref{special_Laplician})
{\small\begin{equation}\label{special_Laplician}L^*=\left[
  \begin{array}{cccc}
    0 & 0 & \cdots & 0 \\
    -p_2 &  p_2 & \cdots & 0 \\
    \vdots & \vdots & \ddots & \vdots \\
    * & \cdots & \cdots & p_{N} \\
  \end{array}
\right],\end{equation}}and let $\bar A_{\rm sys}=I_N\otimes A-L^*\otimes BC$. Note $\bar A_{\rm sys}$ is lower block-triangular. We therefore have $\det(sI-\bar A_{\rm sys})=\det(sI-A)\prod \nolimits_{i=2}^N(sI-A+p_iBC)=\det(sI-A)\alpha_0^{N-1}(s)\prod \nolimits_{i=2}^N\beta_0(s,p_i)$, which should be equal to $\alpha(s)\beta(s,p)$ when $p_1,p_{N+1},...,p_{{ \bar r}}$ are fixed zero. Note that $p_2,...,p_{N}$ are mutually independent.  Hence, \begin{equation}\label{ineq1}\deg \alpha(s)\le n+(N-1)\deg \alpha_0(s).\end{equation}
Inequalities (\ref{ineq1}) and (\ref{ineq2}) indicate  $\deg \alpha(s)= n+(N-1)\alpha_0(s)$, and thus $\deg\beta(s,p)=Nn-\deg \alpha(s)=(N-1)\deg\beta_0(s,p_i)$.
\end{proof}}

{\bf {\emph{Proof of Proposition \ref{new_proposition_nopara}:}}} Suppose that Proposition \ref{new_proposition_nopara} is not true. Then, there exists at least one parameter-dependent mode of $A_{\rm sys}$ that is uncontrollable for all $p\in {\mathbb R}^{{ \bar r}}$ (defined in Lemma \ref{notfixedratio}). Consider the lumped state transition matrix $\bar A_{\rm sys}$ associated with $L^*$ in (\ref{special_Laplician}). From the proof of Lemma \ref{notfixedratio}, $\bar A_{\rm sys}$, obtained by setting $p_1$, $p_{N+1}$,..., $p_{{ \bar r}}$ to be zero, has exactly the same number of parameter-dependent modes as that of $A_{\rm sys}$.\footnote{This condition is vital. Otherwise, setting partial $p_i$ to be zero might transform some parameter-dependent modes to fixed modes, under which circumstance the nonexistence of parameter-dependent uncontrollable modes of the resulting system cannot indicate that the same property holds for the original one.} We will show that, none of the parameter-dependent modes of $\bar A_{\rm sys}$ (which are of one-to-one correspondence to those of $A_{\rm sys}$ when all elements of $p$ are free parameters) can be always uncontrollable, thereby causing a contradiction.

To this end, for each $i\in \{2,...,N\}$, let $\lambda_i$ be a zero of $\beta_0(s,p_i)$. Then, $\lambda_i$ is an eigenvalue of $A-p_iBC$.
Let $\mu_{i}$ be an {\emph{arbitrary}} left eigenvector of $A-p_iBC$ associated with $\lambda_i$. We obtain that $\lambda_i$ and $\mu_{i}$ depend on $p_i$, and for almost all $(p_2,...,p_{N-1}) \in {\mathbb R}^{N-1}$, $\lambda_i\notin \sigma(A)$. Let $w_i$ be one left eigenvector of $L^*$ associated with $p_i$. Note that,  $(w_i^{\intercal}\otimes {\mu_{i}}^{\intercal})\bar A_{\rm sys} = \lambda_i (w_i^{\intercal}\otimes {\mu_{i}}^{\intercal})$,  where $\bar A_{\rm sys}=I_N\otimes A-L^*\otimes BC$, ${\mu_{i}}^{\intercal}(A-p_iBC)=\lambda_i{\mu_{i}}^{\intercal}$, and $w_i^{\intercal}L^*=p_iw_i^{\intercal}$ are used. This means $w_i\otimes\mu_{i}$ is a left eigenvector of $\bar A_{\rm sys}$ associated with $\lambda_{i}$. {On the other hand,  for two distinct $i,j\in\{2,...,N\}$, we know from \citet[Lem. 2]{Rational_function} that $\beta_0(s,p_i)$ and $\beta_0(s,p_j)$ share no common zeros for almost all $(p_i,p_j)\in {\mathbb R}^2$, as $\beta_0(s,p_i)$ and $\beta_0(s,p_j)$ share no common $(s,p)$-factors noting that $p_i$ and $p_j$ are independent.} Hence, $\lambda_i\ne \lambda_j$ for almost all $(p_2,...,p_{N-1})\in {\mathbb R}^{N-1}$.  Suppose that the parameter-dependent mode $\lambda_i$ is always uncontrollable. Then, by the PBH test, for almost all $(p_2,...,p_{N-1})\in {\mathbb R}^{N-1}$ it holds $(w_i^{\intercal}\otimes \mu_i^{\intercal}) ({\bf e}^{N}_1\otimes B)=(w_i^{\intercal}{\bf e}^{N}_1)\otimes (\mu_i^{\intercal}B)=0.$
 As $(L^*,{\bf e}^{N}_1)$ is structurally controllable from Corollary \ref{no_repeated_L}, $w_i^{\intercal}{\bf e}^{N}_1\ne 0$ for almost all $(p_2,...,p_{N-1})\in {\mathbb R}^{N-1}$. To make the previous equation true, it must hold that $\mu_i^{\intercal}B=0$. Notice that by definition, $\mu_i^{\intercal}(\lambda_iI-A+p_iBC)=0.$
We therefore get $\mu_i^{\intercal}(\lambda_iI-A)=0$,
which means that $\lambda_i\in \sigma(A)$ for almost all $(p_2,...,p_N)\in {\mathbb R}^{N-1}$. This contradicts the fact that $\lambda_i$ is a parameter-dependent mode of $\bar A_{\rm sys}$. With the above analysis, we conclude that there cannot exist a parameter-dependent mode of $A_{\rm sys}$ that is always uncontrollable. \hfill $\square$

\end{document}